\documentclass[a4paper,cleveref]{lipics-v2021}
\usepackage[utf8]{inputenc}
\usepackage{algorithm}
\usepackage[noend]{algpseudocode}

\usepackage{amsmath,amsthm,amssymb}
\usepackage{tikz}
\hideLIPIcs
\nolinenumbers

\newcommand{\var}{\mathrm{var}}

\newcommand{\logb}{\mathrm{\log\text{-}bound}}
\newcommand{\vset}[1]{\mathcal{#1}}

\newcommand{\datarule}{{\,:\!\!-\,}}
\usepackage{bm}

\title{Size Bounds for CQs Under Acyclic Constraints}
\author{Stefan Mengel}{Univ.~Artois, CNRS, Centre de Recherche en Informatique de Lens (CRIL), Lens, France}{mengel@cril.fr}{https://orcid.org/0000-0003-1386-8784}
{}
\author{Andrei Romashchenko}{LIRMM Univ Montpellier \& CNRS, France}{andrei.romashchenko@lirmm.fr}{https://orcid.org/0000-0001-7723-7880}{}

\authorrunning{S.~Mengel, A.~Romashchenko}

\keywords{database theory, conjunctive queries, query size bounds, polymatroids, information theory}

\ccsdesc{Theory of computation~Database theory}
\ccsdesc{Mathematics of computing~Information theory}

\begin{document}

\maketitle

\begin{abstract}
We study size bounds for conjunctive query (CQ) 
results which in recent years have played a crucial role in database theory. In particular, we compare the so-called entropic bound which is known to be asymptotically tight but not known to be computable, and its computable relaxation called the polymatroid bound,  which is generally not tight. We focus here on conjunctive queries under acyclic functional dependencies. These queries are known to be well-behaved in the sense that, in the case without projections, both bounds coincide. We show that this picture changes when projections are allowed: in this case, even for acyclic functional dependencies, there is in general a polynomial gap between the two bounds. We complement this negative result by showing a special case for which the polymatroid bound is tight under acyclic functional dependencies and projections that is characterized by the position of the output variables in a topological order of the query variables.
\end{abstract}

\section{Introduction}\label{sec:acyclic}

Computing bounds on the size of query results is crucial in databases. In a database systems context, this task is called \emph{cardinality estimation} and plays an important  part in query optimization, see e.g.~\cite{selinger1979access,leis2015good}. There are different types of heuristics that solve it in practical systems with varying quality, and they have a great influence on query evaluation runtimes, see e.g.~\cite{leis2015good,ZhangMKOS25}.
In contrast to these heuristic estimates in systems work, in database theory the focus is on bounds for conjunctive queries that are provably correct. The idea is to compute bounds on the output size based on more or less detailed statistics of the input database. These statistics generally include the cardinalities of the input tables, but also more refined measures like degree information and functional dependencies (FDs).

Size bounds for query answers are not only intriguing by themselves, but have also played a crucial role for the development of algorithms. For example, the AGM-bound~\cite{atserias2013size} has lead to a revolution in join evaluation by the development of worst-case optimal join algorithms~\cite{ngo2018worst,ngo2018worstb}, i.e.~algorithms whose runtime is linear in the size of the biggest possible query answer. The link between bounds and algorithms is even more direct when taking into account functional dependencies in the data or their generalization degree constraints: given statistics on the input database, the so-called \emph{polymatroid bound} for the output size can be computed by solving a linear program. The solution to this program can then be translated into a sequence of operations that an algorithm, called PANDA~\cite{Khamis0S17}, uses to evaluate the query. The computation of the bound thus guides the algorithm whose runtime turns out to be quasi-linear in the bound.

So to speed up the evaluation of queries, it is important that the runtime bounds are as small as possible while still being correct. For join queries with only cardinality constraints on the input database, it is known that the AGM-bound is tight~\cite{atserias2013size}, so for every query $Q$ and every set of cardinality constraints, there is a database $D$ satisfying the cardinality constraints such that the size of output size $Q(D)$ is up to a constant factor equal to the AGM-bound. This is the underlying reason why we know that the worst-case optimal algorithms, whose runtime is proportional to the AGM-bound, are in fact worst-case optimal.

Once we add more refined constraints, the situation becomes less satisfying: when considering degree statistics and, as an important special case, functional dependencies for the computation of upper bounds, we know that the so-called \emph{entropic bound} is a tight bound on the worst-case size of query results~\cite{Khamis0S17,Suciu23}. While this sounds promising, this result comes with a catch: it is not known if the entropic bound is even computable! This is due to deep questions in information theory related to the lack in good understanding of the so-called entropic cone~\cite{Yeung12}. While there has been progress in that direction over the last few years, see e.g.~\cite{zhang1998characterization,makarychev2002new,matus2007infinitely,dougherty2011non,matuvs2016entropy,csirmaz2025exploring,laszlo2026information} 
and the surveys in \cite{yeung2003first,yeung2015facets,chan2011recent},
fully understanding the situation seems to be far off. In database theory, one consequence of the resulting limited understanding of the entropic bound is that, in contrast to other bounds, so far it does not seem to lead to an algorithm for query evaluation.

One way of side-stepping this problem is considering looser bounds with better algorithmic properties. This is in fact what the already mentioned \emph{polymatroid bound} does. It is an approximation of the entropic bound that can be computed with linear programming and leads to a fast evaluation algorithm for conjunctive queries, the PANDA algorithm. However, it is known that the polymatroid bound is in general not tight, i.e.~the size bounds for query results that it gives can be off by a polynomial factor, even for the restricted case of join queries with functional dependencies. As a consequence, the PANDA algorithm is generally \emph{not} worst-case optimal and there is in fact no known worst-case optimal algorithm in the setting. As a result, the complexity of conjunctive queries in the presence of FDs and more general degree constraints is less well understood than with only cardinality constraints. Moreover, unless there are breakthroughs in information theory, we seem to be stuck with this unsatisfying situation.

However, while in general the situation for bounds on query results under degree constraints might be difficult to understand, we can still aim to better understand it for restricted cases. For example, size bounds are better understood for the case in which all degree constraints are simple, i.e.~the dependency is only on single attributes instead of sets~\cite{GottlobLVV12,GogaczT17,abo2016computing,Suciu23}. 
Another class of constraints that has attracted considerable interest is that of \emph{acyclic} degree constraints. Roughly speaking, these are sets of constraints in which dependencies can only appear in a specific direction (see \Cref{sec:prelim} for a formal definition).
These constraints were introduced in~\cite{ngo2018worstb} where it was shown that the entropic bound and the polymatroid bound coincide when restricted to acyclic degree constraints. This also results in worst-case optimal algorithms for join evaluation~\cite{ngo2018worstb,CapelliIS25}. Moreover, for this class of constraints there are efficient sampling algorithms~\cite{WangT26,CapelliIS25}. 

So far, the above positive results have only been shown for queries without any projections, i.e.~all variables in the query are output variables. In practice however, projection is very common. After all, users might often be in a situation in which they do not care about some of the attributes in tables of a database. For example, when studying a large corpus of SPARQL queries, which is more general than what we study here, projection was found in 14.98\% of all queries~\cite{BonifatiMT20}. From a more theoretical perspective, projection is useful to materialize so-called bags in tree decomposition based approaches~\cite{GottlobLS02,GroheM14,Khamis0S17}. So studying the effect of projections on size bounds %
is certainly worthwhile. 
Concretely, the question we answer here is the following: do the entropic bound and the polymatroid bound for conjunctive queries with acyclic degree constraints still coincide for queries with projections?

We answer this question in the negative by constructing a query with acyclic functional dependencies for which the polymatroid bound exceeds the entropic bound by a polynomial factor. Consequently, the polymatroid bound is not tight for such queries, and existing algorithms are not worst-case optimal. This is in contrast to the setting without projections.

Our way of approaching the problem is as follows: we start with a known query with functional dependencies for which the polymatroid bound is known to be not tight from~\cite{Khamis0S16,Suciu23}. The functional dependencies for this query are highly cyclic, so it does not answer our question directly. The key ingredient of our construction is then that we show of way of ``untangling'' the functional dependencies. To this end, for every variable involved in FDs that violate acyclicity, we break the cycles by adding a new variable that takes its role in the FDs. Doing this systematically and adding some more FDs between old and new variables, leads to a new query with a new set of FDs that are acyclic now. Our result then follows by showing that for the example query from~\cite{Khamis0S17,Suciu23}, this construction changes neither the entropic bound nor the polymatroid bound.

We complement these negative findings by showing that if in a query with acyclic FDs the output variables come ``before'' the projected variables in a well-defined sense, then the entropic bound and the polymatroid bound coincide.

\section{Preliminaries}\label{sec:prelim}

\paragraph*{Conjunctive Queries}

We assume that the reader is familiar with the basics on conjunctive queries, see e.g.~\cite{ArenasBLMP21}. 
We write a conjunctive query $Q$ as a comma-separated list   of the form
\begin{align*}
	Q(\vset X) \datarule  R_1 (\vset X_1),  \ldots, R_\ell (\vset X_\ell) 
\end{align*}
where each $R_i$ is a relation symbol of the database schema and ${\vset X}_i $
is the set of variables occurring in the atom $R_i$. %
The terms $R_i(\vset X_i)$ are called \emph{atoms}.
The variables set $\vset X$ is the set of \emph{output variables}, $\var(Q) = \bigcup\limits_{i=1}^\ell \vset X_i$ denotes the set of all variables occurring in the query
(or simply \emph{the variables set of $Q$}), and finally $\{Y_1, \ldots, Y_m \} = \var(Q)\setminus \vset X$ are the \emph{projected variables}.
As usual, $Q(\vset X)$ is interpreted as a conjunction of atoms, with the existential quantifiers over all projected variables, i.e., 
\begin{align*}
	 \exists  (Y_1, \ldots, Y_m) \ R_1 (\vset X_1) \wedge  \ldots \wedge R_\ell (\vset X_\ell).
\end{align*}

Given a database $D$ that contains a relation $R^D$ for every relation symbol $R_i$, we denote by $Q(D)$ the query result with respect to the usual set semantics. We will mostly be interested in the number of query answers $|Q(D)|$, i.e.~the number of tuples in $Q(D)$.

We will consider functional dependencies of the form $\vset W \rightarrow \vset U$ for $\vset W, \vset U\subseteq \var(Q)$. We require that every such functional dependency is \emph{covered} by an atom, i.e.,~there is an atom $R_i(\vset X_i)$ with $\vset U\cup \vset W \subset \vset X_i$. We say that a database $D$ is consistent with the functional dependency $\vset W \rightarrow \vset U$ if in the covering relation $R_i^D$, for every assignment to the variables in $\vset W$, there is at most one extension to the variables in $\vset U$. The database $D$ is consistent with a set $\Sigma$ of functional dependencies if $D$ is consistent with all functional dependencies in $\Sigma$.

\label{directed-graph}
We assign a directed graph $G_\Sigma$ to every set $\Sigma$ of functional dependencies as follows: the vertices are the variables appearing in $\Sigma$, and for every functional dependency $\vset U \rightarrow \vset W$, the graph $G_\Sigma$ has the edges $\{UW\mid U\in \vset U, W\in \vset W\}$. We then call $\Sigma$ \emph{acyclic} if $G_\Sigma$ is a directed acyclic graph.

\begin{example}\label{ex:acyclic}
	Consider the query 
	\begin{align*}
		Q(X_1,X_2, X_3) \datarule R_1(X_1, X_2, Y_1),R_2(X_1, Y_2, Y_3),R_3(X_2, X_3, Y_3)
	\end{align*}
	with the set of FDs $\Sigma $ contaning $X_1\rightarrow X_2Y_1, X_1Y_2\rightarrow Y_3, X_2 \rightarrow X_3Y_3$. The output variables are $X_1, X_2, X_3$. The graph $G_\Sigma$ the contains the edges $X_1 X_2, X_1Y_1, X_1Y_3, Y_2Y_3, X_2X_3, X_2Y_3$. The FDs are acyclic which can be seen e.g.~by the topological order 
	\[
X_1 \prec X_2  \prec  X_3  \prec   Y_1  \prec  Y_2 \prec  Y_3.
	\] 
\end{example}

\paragraph*{Polymatroids and Entropy}

We only introduce some minimal notions from information theory here. The interested reader is referred to~\cite{Yeung12} for more background.

Fix a finite domain $\vset Y$. Then a \emph{set function} on $\vset Y$ is a function $f:2^{\vset Y} \rightarrow \mathbb R_+$. 
We say that $f$ is \emph{monotone} if for all $\vset A \subseteq \vset B \subseteq \vset Y$ we have $f(\vset A) \le f(\vset B)$. The function $f$ is called \emph{submodular} if for all $\vset A, \vset B \subseteq \vset Y$ we have
\begin{align*}
	f(\vset A\cup \vset B) + f(\vset A\cup \vset B) \le f(\vset A) + f(\vset B).
\end{align*}
Finally, $f$ is called a \emph{polymatroid} if it is monotone and submodular and $f(\emptyset) = 0$. 
We denote by $\Gamma_{\vset Y}$ the set of all polymatroids on a set $\vset Y$.\footnote{We remark that the notation $\Gamma_{\vset Y}$ is slightly non-standard. 
Generally, in the literature, the notation $\Gamma_{n}$ is used for polymatroids on sets of size $n$. We chose our notation here because not only the size of $\vset Y$ will be important to us but also which elements contained in it.}

To streamline notation on polymatroids, we make several conventions: for sets $\{A_1, \ldots, A_\ell\}$, we often write $h(A_1\ldots A_\ell)$ instead of $h(\{A_1, \ldots, A_\ell\})$. Similarly, for sets $\vset A_1, \ldots, \vset A_\ell$, we write $h(\vset A_1 \ldots \vset A_\ell)$ for $h\left(\bigcup_{i\in [\ell]} \vset A_i\right)$. We also mix this notation, e.g.~$h(A_1\vset B)$ stands for $h(\{A_1\}\cup \vset B)$ and similar.

One particular interesting class of polymatroids are entropic functions which we define next. Let $Y$ be a random variable on a finite\footnote{Entropy can also be defined for continuous random variables with infinite domains. However, we will only restrict to finite domains since those are sufficient for our purposes.} domain $D$. Then the \emph{entropy} of $Y$ is 
\begin{align*}
	H(Y) = - \sum_{d\in D} \Pr(Y=d) \log(\Pr(Y=d)) 
\end{align*}
with the usual convention $  \Pr(Y=d) \log(\Pr(Y=d))=0$ if $ \Pr(Y=d) = 0$.
Note that, since $0\le \Pr(Y=d) \le 1$, we have that $H(Y)$ is always non-negative.
The entropy of $Y$ has an interpretation as the amount of information that one can get from sampling $Y$, but we do not need anything beyond some properties of entropy here.

Given a set $\vset Y$ and random variables $\{Y_i \mid i \in \vset Y\}$, each with domain $D$, we define a set function $h :2^{\vset Y} \rightarrow \mathbb R_+$ as follows: for every set $\vset A = \{i_1, \ldots, i_\ell \}\subseteq \vset Y$, let $Y_{\vset A}$ be the random variable with domain $D^\ell$ with $\Pr(Y_{\vset A}= (d_1, \ldots d_\ell)) = \Pr\left(Y_{i_1}= d_1, \ldots, Y_{i_\ell} = d_\ell\right)$. Then $h(\vset A) := H(Y_{\vset A})$. Functions $h$ which can be represented that way are called \emph{entropic}. We denote by $\Gamma^*_{\vset Y}$ the set of all entropic functions over the set $\vset Y$.

The following result goes back to Shannon's foundational work on information theory.
\begin{theorem}\cite{Shannon48}\label{thm:shannon}
	All entropic functions are polymatroids.
\end{theorem}
In some cases (see Section~\ref{sct:entropicbound}) we abuse notation by allowing random variables (rather than their indices) as arguments of the set function $h$. 
Thus, given a set of random variables $\{Y_i \mid i \in [n]\}$, we write $h(Y_1,Y_2)$ for $h(\{1,2\})=H(Y_1,Y_2)$. 
Similarly, given random variables $A,B,X,Y$, we write $h(A,B,X,Y)$ for $H(A,B,X,Y)$.

\begin{definition}\label{def:conditional}
Given a polymatroid $h$ and not necessarily disjoint subsets $\vset A, \vset B$, we define $h(\vset A\mid \vset B) := h(\vset A\vset B) - h(\vset B)$.
\end{definition}
\begin{remark}
The definition of $h(\vset A\mid \vset B)$ is motivated by information theory, where the value $h(\vset A\mid \vset B)$ is called the \emph{conditional entropy of $\vset A$ given $\vset B$}.
In matroid theory, a similar transformation of the rank function appears in the contraction of (poly)matroids.
	\end{remark}

The following easy observation will be useful throughout this paper.

\begin{observation}
\label{obs:chainrule}
	Let $h$ be a polymatroid on the set $\vset Y$ and let $\vset A, \vset B, \vset C\subseteq Y$. Then 
	\begin{align*}
		h(\vset A\vset B|\vset C) = h(\vset B|\vset C) + h(\vset A|\vset B \vset C).
	\end{align*}
\end{observation}
\begin{proof}
	Applying Definition~\ref{def:conditional}, we get
	\begin{align*}
		h(\vset B|\vset C) + h(\vset A|\vset B \vset C) &= h(\vset B \vset C) - h(\vset C) + h(\vset A \vset B \vset C) - h(\vset B \vset C)\\
		&= h(\vset A\vset B \vset C) - h(\vset C) \\
		&= h(\vset A\vset B |\vset C).\qedhere
	\end{align*}
\end{proof}
\begin{observation}\label{obs:moreconditioned}
	Let $\vset A, \vset B, \vset C$ be sets of variables and let $h$ be a polymatroid. Then
	\begin{align*}
		h(\vset A\mid \vset B\vset C) \le h(\vset A \mid \vset B).
	\end{align*}
\end{observation}
\begin{proof}
	By submodularity, we have
	\begin{align*}
		h(\vset A \vset B \vset C) + h(\vset B) \le h(\vset A\vset B) +h(\vset B\vset C).
	\end{align*}
	Rearranging according to Definition~\ref{def:conditional} directly yields the claim.
\end{proof}

\paragraph*{Size Bounds for Query Results}

Entropic functions play a key role in size bounds for conjunctive queries under functional dependencies. To this end, consider a CQ 
\[ Q(\vset X) \datarule R_1(\vset X_1), \ldots,R_\ell(\vset X_\ell). \] Let $\vset Y = \bigcup\limits_{i=1}^\ell \vset X_i$.
Throughout the paper, we will let $\vset X$ be the set of output variables and $\vset Y$ be all variables of the query.
 We assume that we have a set $\Pi$ that contains for every relation $R_i$ a cardinality constraint $N_i$ on its size. Moreover, we have a set $\Sigma$ of FDs. We say that a set function $h:2^{\vset Y}\rightarrow \mathbb R_+$ \emph{satisfies} $\Pi$ and $\Sigma$, in symbols $h\models (\Pi, \Sigma)$ if for every $\vset X_i$ we have $h(\vset X_i) \le \log(N_i)$ and for every FD $\vset U \rightarrow \vset V$ in $\Sigma$, we have $h(\vset V|\vset U) = 0$.

Let $K$ be a class of set functions over $\vset Y$. Then we define the quantity 
\begin{align*}
	\logb_K(\vset X, \Gamma,\Sigma) := \sup_{h\in K: h\models (\Pi, \Sigma)} h(\vset X).
\end{align*}
In what follows we discuss two important special cases of this bound: $K = \Gamma^*_{\mathcal Y}$ and $K=\Gamma_{\mathcal Y}$.

\begin{theorem}[Entropic Bound~\cite{GottlobLVV12}]
\label{th:5}
	Let $Q(\vset X)$ be a conjunctive query with the set of variables~$\vset Y$, with cardinality constraints $\Pi$ and FDs $\Sigma$.
	If $D$ is a database with $|R_i^D| \le N_i$ for every atom $R_i$ of $Q(\vset X)$, then
	\(
		|Q(D)| \le 2^{\logb_{\Gamma^*_{\vset Y}}(\vset X, \Pi, \Sigma)}.
	\)
\end{theorem}

It is known that the entropic bound is tight in an asymptotic sense~\cite{GogaczT17}. However, unfortunately, it is not known if $\logb_{\Gamma_{\vset Y}}(\vset X, \Pi, \Sigma)$ is even computable, so already in~\cite{GottlobLVV12}, the following approximation of the entropic bound was formulated.

\begin{theorem}[Polymatroid Bound~\cite{GottlobLVV12}]
\label{th:6}
	Let $Q(\vset X)$ be a conjunctive query  with the set of variables $\vset Y$, with cardinality constraints $\Pi$ and FDs $\Sigma$.
	 If $D$ is a database with $|R_i^D| \le N_i$ for every atom $R_i$ of $Q(\vset X)$, then
	\(
		|Q(D)| \le 2^{\logb_{\Gamma_{\vset Y}}(\vset X, \Pi, \Sigma)}.
	\)
\end{theorem}
Note that the only difference between the entropic bound and the polymatroid bound is that in the former we take the supremum over all entropic functions while in the latter we maximize over the strictly bigger set of all polymatroids. As a consequence, we get that we always have $\logb_{\Gamma^*_{\vset Y}}(\vset X, \Pi, \Sigma) \le \logb_{\Gamma_{\vset Y}}(\vset X, \Pi, \Sigma)$.

\section{Separating the Entropic and Polymatroid Bound}

In this section we will show our main result that in general the entropic bound can be smaller by a polynomial factor than the polymatroid bound even for the case of acyclic FDs.

\begin{theorem}\label{thm:acyclicgap}
	There is a conjunctive query $Q(\vset X)$ with variable set $\vset Y$, 
	a set of cardinality constraints $\Pi$ and a set of acyclic FDs $\Sigma$ such that 
	\[
	\logb_{\Gamma^*_{\vset Y}}(\vset X, \Pi, \Sigma) < \logb_{\Gamma_{\vset Y}}(\vset X, \Pi, \Sigma).
	\] 
\end{theorem}

We will show Theorem~\ref{thm:acyclicgap} in the remainder of this section.
The general idea is that we start from an instance with \emph{cyclic} constraints for which we already know that it has a gap between the polymatroid bound and the entropic bound (Section~\ref{sec:introduceZY}). 
Then, we transform this instance in an example with similar properties but \emph{acyclic} constraints. To this end, 
we introduce copies of variables to break the cycles (Section~\ref{sec:untangling}). 
To prove Theorem~\ref{thm:acyclicgap}, 
it suffices to show that untangling the cyclic constraints in a query preserves both the polymatroid bound (we prove this Section~\ref{sct:polymatroid-bound}) and the entropic bound (we prove this in Section~\ref{sct:entropicbound}).

\subsection{Introducing the base instance}\label{sec:introduceZY}

In this section, we begin the proof of Theorem~\ref{thm:acyclicgap} by introducing the query that in~\cite{Khamis0S17,Suciu23} was used to show a gap between the entropic bound and the polymatroid bound. We will use this query to as the basic building block of for the proof of Theorem~\ref{thm:acyclicgap}, as discussed before.
The query $Q_{ZY}$ from from~\cite{Khamis0S17, Suciu23}, called the \emph{Zhang--Yeung query} in~\cite{Khamis0S17}, is defined as follows:
\begin{equation}
\begin{aligned}
	Q_{ZY}(A,B,X,Y,C) \datarule &R_1(X,Y),  R_2(A,X),  R_3(A,Y),\\ &R_4(B,X), R_5(B,Y),  R_6(C), R_7(A,B,X,Y,C).\label{eq:basequery}
\end{aligned}
\end{equation}
The variable set of this query is $\vset Y = \{A,B,C,X,Y\}$.
As in~\cite{Khamis0S17,Suciu23}, we add the cardinality constraints $|R_i|\le N^3$ for $i\in [5]$ and $|R_6|\le N^2$ and the functional dependencies 
\begin{align}
	AB \rightarrow XYC,\\
	AXY \rightarrow BC,\\
	\label{a1} BXY \rightarrow AC,\\
	AC \rightarrow BXY,\\
	\label{a2} XC \rightarrow ABY,\\
	\label{a3} YC \rightarrow ABX
\end{align}
Clearly, these constraints are highly cyclic, but we will take care of this later. 
We transform the cardinality constraints into the constraints
\begin{align}
	\begin{split}
	&h(XY)\le 3\log(N),
	h(AX)\le 3\log(N),
	h(AY)\le 3\log(N),\\
	&h(BX)\le 3\log(N),
	h(BY)\le 3\log(N),
	h(C)\le 2\log(N).
	\end{split}\label{eq:cczy}
\end{align}
The functional dependencies are transformed into the constraints
\begin{align}
	\begin{split}
	&h(XYC|AB)\le 0,
	h(BC|AXY)\le 0,
	h(AC|BXY)\le 0,\\
	&h(BXY|AC)\le 0, 
	h(ABY|XC)\le 0, 
	h(ABX|YC)\le 0.
	\end{split}\label{eq:fdzy}
\end{align}

In~\cite{Suciu23,Khamis0S17}, it was shown that the polymatroid bound for this query and the size constraints and FDs is $4\log(N)$. 
However, the entropic bound is stronger: it implies for this query the bound at most $\frac{43}{11}\log(N)$, 
as shown in in  \cite{Khamis0S17}. Let us recall how this inequality for the entropic bound was proven in  \cite{Khamis0S17}. 
The crucial ingredient of the argument is the following fact from information theory.
\begin{lemma}[see \cite{zhang1998characterization}]\label{lem:zy} For every entropic function $h$ in variables $A,B,X,Y$, we have
\begin{align*}\label{eq:non-shannon}
	0 \ge & 2 h(X) + 2 h(Y) + h(A) - 3 h(XY) - 3 h(XA) - 3 h(YA) - h(XB) \\ &- h(YB)+ h(AB) + 4 h(XYA) + h(XYB).
\end{align*}
\end{lemma}
One can  combine the inequality from Lemma~\ref{lem:zy} with the constraints in~(\ref{eq:cczy}) and (\ref{eq:fdzy}),
and all polymatroid inequalities (monotonicity and submodularity), and then maximize the value of $h(XYABC)$ compatible with those constraints.
We obtain a linear program for which the optimal value can be computed by a symbolical computation
(by hand or with help of computer). The computed optimal value is the value $\frac{43}{11}\log(N)$ claimed above, see~\cite{Khamis0S17,Suciu23}.
Thus the entropic bound for this query is at most $\frac{43}{11}\log(N)$
(it remains possible that with a better understanding of the laws of Shannon entropy one could get an even tighter bound.)

\subsection{Untangling the functional dependencies}\label{sec:untangling}

\begin{algorithm}[t]
	\caption{The untangling algorithm}
	\label{alg:join}
	\begin{algorithmic}[1] %
		\Procedure{Untangle}{$Q, \vset X, \Pi, \Sigma, \pi$} 
		\State /* we assume $Q$ is a query with output variables $\vset X$ and without projections\ */
		\State $Q' \gets Q$
		\State $\vset{X}' \gets \vset X$
		\State $\vset{Y}' \gets \vset X$
		\State $\Sigma' \gets \Sigma$
		\ForAll{variables $X$ from $\vset X$ in order $\pi$}
			\If{there is an offended FD for $\pi$ via $X$ in $\Sigma'$}
				\State{introduce new variable $X'$}%
				\State{$\vset X' \gets  (\vset  X' \setminus \{X\})\cup \{X'\}$ }
				\State{$\vset Y' \gets  \vset  Y' \cup \{X'\}$ }
				\State{$Q' \gets Q', R_X(X, X')$} %
				\State{add $X\rightarrow X'$ to $\Sigma'$}
				\ForAll{FDs $\vset W\rightarrow \vset V$ in $\Sigma'$ offended for $\pi$ via $X$}
					\State{delete $\vset W\rightarrow \vset V$ from $\Sigma'$}
					\State{add $(\vset W\rightarrow  \left( (\vset V\setminus \{X\})\cup \{X'\} \right)$  to $\Sigma'$}
				\EndFor
			\EndIf
		\EndFor
		\State /* by construction, $Q'(\vset X)$ is a query  with variable set $\vset Y'$ and output variables $\vset X'$ \ */
		\State \Return $Q'(\vset X'), \Sigma'$  %
		\EndProcedure
	\end{algorithmic}
\end{algorithm}

For our purposes, the query $Q_{ZY}$ from~(\ref{eq:basequery}) is not immediately useful since its constraints are cyclic. 
So we transform it into another query $Q_{ZY}'$ with new functional dependencies~$\Sigma'$. 
Rather than presenting an ad-hoc trick, we explain a general procedure 
 that transforms any given query $Q$ without projected variables but with potentially cyclic FDs $\Sigma$ into another query $Q'$ with acyclic FDs $\Sigma'$ but potentially with projections. So essentially in the transformation we gain acyclicity of the FDs at the price of introducing projections. We then show that the polymatroid bound of $Q'$ and $\Sigma'$ is at least that of $Q$ and $\Sigma$. Moreover, applying this procedure on $Q_{ZY}$ and $\Sigma$ yields $Q_{ZY}'$ and $\Sigma'$ for which we can then show the required inequality for the entropic bound. This will prove the main result of this paper.

Consider any query $Q$ without projected variables with a set $\Sigma$ of potentially cyclic functional dependencies. We fix an order $\pi$ of 
the variable set $\vset Y$ of $Q$. Essentially, $\pi$ is a guess for a topological order of the variables. We say that a constraint is \emph{offended 
for the order $\pi$ (via a variable $V$)} if there are variables $V, W$ such that $V$ is before~$W$ in $\pi$ but $W$ is in the body of this dependency  and $V$ is in its head. Clearly, there are no offended constraints for $\pi$ if and only if $\pi$ is a topological order of the dependency graph of $Q$. So the idea is to systematically \emph{untangle} all offended constraints by adding new variables, and eventually end up with a topological order witnessing that the constraints are acyclic. To do so, we are using the procedure given in \Cref{alg:join} which we describe next.

Given a query $Q(\vset X)$ with constraints $\Pi, \Sigma$ and a variable order $\pi$, we go through the list of variables in the order $\pi$ and do the following: if the current variable $X$ is not in the head of an offended constraint for the current order, we skip it and proceed with the next variable. Otherwise, we introduce a  new variable $X'$ (which is essentially a ``copy'' of $X$) and make the link between the original variable $X$ and its copy $X'$
 by a new functional dependency $X\rightarrow X'$. Since for each functional dependency there must be an atom that contains all of its variables, we also add an atom $R_X(X, X')$ where $R_X$ is a fresh relation name. 
 Moreover, in the set of output variables we replace the old variable $X$ by its copy $X'$.
Observe that due to this, there are projections in the new query (even though the original query does not have any), since $X$ is not an output variable anymore.
Finally, for every  occurrence of $X$ in the head of an offended constraint, we replace  $X$ by~$X'$. This completes the construction of the new instance $Q'$ and the set of FDs. 

\begin{observation}
	For every input $(Q, \vset X, \Pi, \Sigma, \pi)$, 
	\Cref{alg:join}  returns $Q',\Sigma'$ such that 
	the obtained set of functional dependencies $\Sigma'$ is acyclic.
\begin{proof}[Proof (sketch)]
	The acyclicity condition can be violated only in the presence of offended constraints. 
	When processing each variable in $\vset X$, the algorithm eliminates every constraint  offended via that variable. 
	As for the newly introduced variables, they appear only in the heads of functional dependencies and therefore cannot violate the acyclicity condition.
\end{proof}
\end{observation}

Let us apply the Algorithm~\ref{alg:join} to  the example~$Q_{ZY}$ presented in Section~\ref{sec:introduceZY}.
We fix the variable order $\pi= ABXYC$. When running the algorithm, we first consider the variable~$A$ for which the constraints (\ref{a1}), (\ref{a2}) and (\ref{a3}) are offended. So we introduce a new variable $A'$ and swap the offended constraints out for 
\begin{align*}
	BXY\rightarrow A'C\\
	XC \rightarrow A'BY\\
	YC \rightarrow A'BX.
\end{align*}
Also, $A'$ becomes an output variable instead of $A$, and we add a new functional dependency $A\rightarrow A'$. 
Going through the list of all variables, we obtain the query 
\begin{equation}
\begin{aligned}
Q_{ZY}'(A', B', X', Y', C) \datarule &R_1(X,Y), R_2(A,X), R_3(A,Y),\\ & R_4(B,X), R_5(B,Y), R_6(C), R_7(A,B,X,Y,C),\\
& R_A(A, A'), R_B(B, B'), R_X(X, X'), R_Y(Y, Y') 
\label{eq:basequery-prime}
\end{aligned}
\end{equation}
with the extended variable set $\vset Y'= \{A,B,X,Y,C, A', B', X', Y'\}$.
The new set of functional dependencies $\Sigma'$ consists of those in $\Sigma$ in which potentially some variables are swapped out, 

\begin{equation}
	\begin{split}
	&AB \rightarrow XYC\\
	&AXY \rightarrow B'C\\
	&BXY\rightarrow A'C\\
	&AC \rightarrow B'X'Y'\\
	&XC\rightarrow A'B'Y'\\
	&YC\rightarrow A'B'X',
	\end{split}\label{eq:untangled}
\end{equation}
and the additional dependencies we added when introducing new variables
\begin{align}
	A \rightarrow A', B \rightarrow B', X \rightarrow X', Y \rightarrow Y'.\label{eq:unaryUntangle}
\end{align}
We keep the same cardinality constraints $\Pi$ as before.

\Cref{thm:acyclicgap} is a direct corollary of the following result, where $\vset Y', \vset X', \Sigma', \Pi$ are taken  from the query $Q'_{ZY}$ as defined above.
\begin{proposition}\label{prop:concretebounds}
	$\logb_{\Gamma^*_{\vset Y'}}(\vset X', \Pi, \Sigma') \le \frac{43}{11 }\log(N)< 4\log(N) \le \logb_{\Gamma_{\vset Y'}}(\vset X', \Pi, \Sigma')$.
\end{proposition}

One  can prove \Cref{prop:concretebounds} with an assistance of a computer, as follows.
First of all, we combine the inequalities for the FDs (\ref{eq:untangled})  and (\ref{eq:unaryUntangle}), the constraints from~(\ref{eq:cczy}), and all polymatroid inequalities and Theorem~\ref{lem:zy},
 in one linear program (LP),  feed this optimization problem into any LP solver,
 and compute the maximal value for $h(X'Y'A'B'C)$. 
This computation gives the polymatroid bound, which in this case is $4$. 
Further,  if  we  include in this linear program  one more constraint --- the inequality from~\Cref{lem:zy} --- 
we obtain a smaller maximal value for $h(X'Y'A'B'C)$, which in this case is $\frac{43}{11}$.
This mechanical computation concludes the proof of the proposition.
In the appendix (see p.~\pageref{sct:appendix}) we discuss this approach in more detail.
However, a proof relying on heavy computation performed by a computer (calculating the optimal value of two huge linear programs)  
might look somewhat unsatisfactory.  So in the following subsections we provide a fully human-checkable argument.

\subsection{The Polymatroid Bound}\label{sct:polymatroid-bound}

In this section, we show the lower  bound on $\logb_{\Gamma^*_{\vset Y'}}( {\vset X'}, \Pi, \Sigma)$ from Proposition~\ref{prop:concretebounds}.  
More generally, we will show that the construction of \cref{alg:join} never decreases the polymatroid bound.
Let us formulate this statement more precisely.

\begin{proposition}
	Let $Q(\vset X)$ be a CQ in variables $\vset X$ without projected variables with cardinality constraints $\Pi$ and functional dependencies $\Sigma$. 
	We apply  \Cref{alg:join} to the input $Q(\vset X), \Pi, \Sigma$  and any variable order $\pi$.
	Let $Q'(\vset X'), \Sigma'$ be the output returned by the algorithm,
	and let $\vset Y'$ be the variables set of the output query $Q'$. Then 
	\begin{align*}
		\logb_{\Gamma_{\vset X}}(\vset X, \Pi, \Sigma) \le \logb_{\Gamma_{\vset Y'}}(\vset X', \Pi', \Sigma').
	\end{align*}
\end{proposition}
\begin{proof}
	We assume that there is polymatroid $h$  in $\vset X$ that respects all cardinality constraints $\Pi$ and functional dependencies $\Sigma$. 
	By definition, the value of $h({\vset X})$ must be below the bound %
	$\logb_{\Gamma_{\vset X}}(\vset X, \Pi, \Sigma)$. 
	We will construct another polymatroid $h'$, defined on ${\vset Y}'$ and respecting the constraints from $\Pi$ and $\Sigma'$. 
	Moreover, in our construction we will have $h'({\vset X}') = h({\vset X})$.
	Thus, the bound $\logb_{\Gamma_{\vset Y'}}(\vset X', \Pi', \Sigma')$ applies in particular to $h'$ and, therefore, to $h$. This implies the statement of the lemma.
	
	To implement the plan sketched above,  
	we define the inverse operation of the variable addition in \Cref{alg:join}. For every variable $X\in \vset Y'$, let $s(X)$ stand for the variable $V\in \vset Y$ such that 
	 $X=V$ (in the case $X\in \vset X$) or $X= V' $  (in the case $X\in  \vset Y'\setminus \vset Y$). Essentially, copies of variables are mapped to their original, while the originals are mapped to themselves.
	 For every variable set $\vset A\subseteq \vset Y'$, we extend this definition to $s(\vset A):= \{s(X)\mid X\in \vset A\}$.
	
	The operation $s(\cdot)$ plays well with the usual operations on sets:
	\begin{claim}
		For all $\vset A,\vset B\subseteq \vset Y'$ we have:
		\begin{itemize}
			\item if $\vset A\subseteq \vset B$, then $s(\vset A)\subseteq s(\vset B)$ (monotonicity),
			\item $s(\vset A\cup \vset B) = s(\vset A)\cup s(\vset B)$, and
			\item $s(\vset A\cap \vset B)\subseteq s(\vset A)\cap s(\vset B)$.
		\end{itemize}
	\end{claim}
	\begin{claimproof}
		The first and second claims are immediately clear from the definition.
		
		For the third claim, note that
		\begin{align*}
			s(\vset A\cap \vset B) = \{s(X)\mid X\in \vset A\cap \vset B\} \subseteq \{s(X)\mid X\in \vset A\} \cap \{s(X)\mid X\in \vset B\} = s(\vset A) \cap s(\vset B).
		\end{align*}
	\end{claimproof}
	Now assume that $h$ is a polymatroid on $\vset X$ satisfying the constraints in $\Pi$ and $\Sigma$. We define a set function $h'$ on $\vset Y'$ as
	 \(
	h'(\vset A):= h(s(\vset A)).
	\)
	Clearly, for the set $\vset X$ we have $h'(\vset X') = h(\vset X)$, since $s(\vset X') = X$. 
	So to prove the lemma it remains to show that $h'$ is a polymatroid that satisfies the constraints of $\Sigma'$. 
	
	We first show that $h'$ is a polymatroid. For monotonicity, consider $\vset A, \vset B\subseteq \vset Y'$ with $\vset A\subseteq \vset B$, then
	\begin{align*}
		h'(\vset A) = h(s(\vset A)) \le h(s(\vset B)) = h'(\vset B), 
	\end{align*}
	where we use the monotonicity of $s$ and $h$. To show that $h'$ is submodular, consider arbitrary sets $\vset A, \vset B\subseteq \vset Y'$. Then
	\begin{align*}
		h'(\vset A\cup \vset B) + h'(\vset A\cap \vset B) &= h(s(\vset A\cup \vset B)) + h(s(\vset A\cap \vset B))\\
		&\le h(s(\vset A)\cup s(\vset B)) + h(s(\vset A)\cap s(\vset B))\\
		& \le h(s(\vset A)) + h(s(\vset B)) &\ {[\text{since $h$ is submodular}]}\\
		& = h'(\vset A) + h'(\vset B).
	\end{align*}
	So, noting that $h'(\emptyset) = h(s(\emptyset)) = h(\emptyset) = 0$, we conclude that $h'$ is a polymatroid.
	
	It remains to show that $h'$ satisfies the functional dependencies in $\Sigma'$. For all variables $X\in \vset  X$ and their copy $X'$, we have $s(X') = s(X) = X$, so $h'(XX') = h(s(XX'))= h(X) = h'(X)$, and thus $h'$ satisfies the unary functional dependencies the algorithms introduced.
	
	Let $\vset W \rightarrow \vset U$ be any other constraint of $\Sigma'$. 
	Then by construction we have $s(\vset W) =\vset W$ and $h'(\vset U) = h(s(\vset U))$. Since $\vset W\rightarrow s(\vset U)$ is a constraint in $\Sigma$, we have $h(s(\vset U)\vset W)= h(\vset W)$ and thus
	\begin{align*}
		h'(\vset U \vset W) = h(s(\vset U \vset W)) = h(s(\vset U)\vset W) = h(\vset W) = h(s(\vset W)) = h'(\vset W).
	\end{align*}
	Therefore, $h'$ respects the the functional dependency $\vset W\rightarrow \vset U$.
\end{proof}

\subsection{The Entropic Bound}\label{sct:entropicbound}

In this subsection we will show that the entropic bound for the query $Q_{ZY}'$ with the FDs we constructed in Section~\ref{sec:untangling} is at most $\frac{43}{11}\log(N)$.
As we have shown in Section~\ref{sct:polymatroid-bound}, this bound is strictly less than the polymatroid bound. Our starting point is the proof of an upper bound for $Q_{ZY}$ in~\cite{Khamis0S17,Suciu23},
where it was shown that for all entropic functions $h$ in variables $(A,B,C,X,Y)$ the following inequality is valid:
\begin{align}
	\begin{split}
	11 h(ABXYC) \le &3 h(XY) + 3h(AX) + 3h(AY) +h(BX) + h(BY) + 5h(C) \\
	& + h(XYC|AB) + 4h(BC|AXY) + h(AC|BXY) \\&+ h(BXY|AC)+ 2h(ABY|XC) + 2h(ABX|YC).
	\end{split}\tag{*}\label{eq:suciu}
\end{align}
This inequality follows from a combination of Theorem~\ref{lem:zy} with several polymatroid inequalities. 
In the application of (\ref{eq:suciu}) discussed in~\cite{Khamis0S17,Suciu23}, all terms on the right-hand side are bounded,
so the inequality implies an upper bound for $h(ABXYC)$.
This establishes a separation between the entropic bound and the polymatroid bound in the setting considered in~\cite{Khamis0S17,Suciu23}.

We cannot use  (\ref{eq:suciu}) directly, since in our settings we do not have a bound on the terms appearing on the right-hand side. 
Moreover, we need to bound not the value of $h(ABXYC)$ but  $h(A'B'X'Y'C)$ 
(for $A',B',X',Y'$ that are defined in Algorithm~\ref{alg:join}).
This is we need to prove that the following inequality inspired by (\ref{eq:suciu}) is true for all entropic functions $h$:
\begin{align*}
	11 h(A'B'X'Y'C) \le &3 h(XY) + 3h(AX) + 3h(AY) +h(BX) + h(BY) + 5h(C) \\
	& + h(XYC|AB) + 4h(B'C|AXY) + h(A'C|BXY) \\
	&+ h(B'X'Y'|AC)+ 2h(A'B'Y'|XC) + 2h(A'B'X'|YC).\tag{**}\label{eq:suciu-modified}
\end{align*}
With (\ref{eq:suciu-modified}), we can show that $h(A'B'X'Y'C)$ is bounded by $\frac{43}{11}\log(N)$,
which proves the upper bound for $\logb_{\Gamma_{\vset Y'}}(\vset X', \Pi, \Sigma')$  in Proposition~\ref{prop:concretebounds}.
Indeed, we observe that for $h$ satisfying the FDs in $\Sigma'$ (from  Proposition~\ref{prop:concretebounds}),
each term of conditional entropy in (\ref{eq:suciu-modified})  is equal to zero, due to the  functional dependencies in $\Sigma'$, see \eqref{eq:untangled}.
Each term of  non-conditional entropy on the right-hand side of (\ref{eq:suciu-modified}) is bounded by the cardinality constraints in $\Pi$, see \eqref{eq:cczy}. 
Combining these bounds we obtain $h(A'B'X'Y'C) \le \frac{43}{11}\log(N)$, as required.

It remains to prove (\ref{eq:suciu-modified}) based on  (\ref{eq:suciu}). 
It is not very hard to show that $h(A'B'X'Y'C)\le h(ABXYC)$ under the condition that $h$ is a polymatroid that satisfies $\Pi$ and $\Sigma'$, so the left-hand side of (\ref{eq:suciu-modified}) is smaller than that of (\ref{eq:suciu}). 
However, this does not directly help us since the quantities of the conditional entropies in (\ref{eq:suciu-modified}) can also be smaller than those in (\ref{eq:suciu}). 
To resolve this issue, we show that when going from (\ref{eq:suciu}) to (\ref{eq:suciu-modified}) the decrease in  value on the left-hand side cannot be bigger than on the right-hand side.  
The main technical argument is the following lemma which lets us compute the change in conditional entropies when substituting a variable $V$ by $V'$ in \Cref{alg:join}.

\begin{lemma}\label{lem:Atoa}
	Let $h$ be a polymatroid on a domain $\vset Y$ and let $\vset L, \vset M\subseteq \vset Y$ be disjoint. 
	Let moreover $Z, Z'\in \vset Y$ elements not in $\vset L\cup \vset M$ such that $h(Z'|Z) = 0$. Then 
	\begin{align*}
	h(Z\vset L|\vset M) - h(Z'\vset L|\vset M) = h(Z |  Z' \vset L\vset M ).
	\end{align*}
\end{lemma}

\begin{proof}
	By Observation~\ref{obs:chainrule} %
	we have  
	\begin{align*}
		h(ZZ'\vset L|\vset M) = h(Z|Z' \vset L \vset M ) + h(Z' \vset L | \vset M).
	\end{align*}
	Further, from Observation~\ref{obs:chainrule} and Observation~\ref{obs:moreconditioned} we obtain 
	\begin{align*}
		h(ZZ' \vset L|\vset M) = h(Z'|Z\vset L\vset M) + h(Z\vset L |\vset M) = h(Z\vset L |\vset M).
	\end{align*}
	So we get 
	\(
		h(Z\vset L|\vset M) - h(Z'\vset L|\vset M) = h(Z | Z' \vset L\vset M ),
	\)
	and the lemma follows.
\end{proof}

Another ingredient that we need is the following.

\begin{lemma}\label{lem:Atoa-otherside} Let $h$ be a polymatroid on a domain $\vset Y$ and let $\vset L, \vset M\subseteq \vset Y$ be disjoint. Let moreover $Z, Z'\in \vset Y$ elements not in $\vset L\cup \vset M$ such that $h(Z'|Z) = 0$. Then 
	\(
	h(\vset L|Z \vset M) \le h(\vset L|Z' \vset M).
	\)	
\end{lemma}

\begin{proof}
	Expanding the definition of conditional entropy, we get 
	\begin{align*}
		h(\vset X|A \vset Y) &= h(A\vset X  \vset Y) - h(A \vset Y)\\
		h(\vset X|A' \vset Y) &= h(A'\vset X  \vset Y) - h(A' \vset Y).
	\end{align*}
	Moreover, by Lemma~\ref{lem:Atoa}
	\begin{align*}
		h(A\vset X \vset Y) - h(A' \vset X  \vset Y) = h(A|A' \vset X \vset Y)\\
		h(A \vset Y) - h(A' \vset Y) = h(A| A' \vset Y ).
	\end{align*}
	So we obtain 
	\begin{align*}
		h(\vset X | A' \vset Y) - h(\vset X | A \vset Y) 
		&= h(A' \vset X  \vset Y) - h(A' \vset Y) -  h(A \vset X \vset Y) + h(A \vset Y)\\
		&= - h(A|A' \vset X \vset Y ) + h(A| A' \vset Y  ) \\
		&\ge 0,  \ \text{[by Observation~\ref{obs:moreconditioned}]}\\
	\end{align*}
	and we are done.
\end{proof}

In what follows we apply Lemma~\ref{lem:Atoa} and ~\ref{lem:Atoa-otherside} to entropic functions which, by Theorem~\ref{thm:shannon} are a special class of polymatroids.

We now come back to the proof of (\ref{eq:suciu-modified}). We begin with \eqref{eq:suciu} and substitute step-by-step the occurrences of $A,B,X,Y$ by $A',B',X',Y'$, as in the untangling construction of \Cref{alg:join}. 
In the \emph{first step} we substitute $A$ by $A'$ in several occurrences in \eqref{eq:suciu}; 
this affects the terms  $h(ABXYC)$, $h(AC|BXY)$, $h(ABY|XC)$, $h(ABX|YC)$ shown in bold below:
\begin{align}
	\begin{split}
	\bm{ 11 h(ABXYC) } \le &3 h(XY) + 3h(AX) + 3h(AY) +h(BX) + h(BY) + 5h(C) \\
	& + h(XYC|AB) + 4h(BC|AXY) + \bm{h(AC|BXY)} \\&+ h(BXY|AC)+ \bm{ 2h(ABY|XC) } +  \bm{ 2h(ABX|YC) }.
	\end{split}\nonumber
\end{align}
By Lemma~\ref{lem:Atoa}, to balance each  substitution, we have to add the value of $h(A|A'BXYC)$.
Summing the coefficients,  we see that we have to add this quantity $11$ times on the left-hand side and $5$ times on the right-hand side.
Since $h$ is non-negative, we may therefore conclude that 
\begin{align}
	11 h(A'BXYC) \le &3 h(XY) + 3h(AX) + 3h(AY) +h(BX) + h(BY) + 5h(C) \nonumber\\
	& + h(XYC|AB) + 4h(BC|AXY) + h(A'C|BXY) \nonumber \\&+ h(BXY|AC)+ 2h(A'BY|XC) + 2h(A'BX|YC).\label{eq:suciu2}
\end{align}
In the \emph{second step}, we substitute $B$ by $B'$ in the terms $h(A'BXYC)$, $h(BC|AXY)$, $h(BXY|AC)$, $h(A'BY|XC)$, $h(A'BX|YC)$
shown in bold below:
\begin{align}
	\bm{11 h(A'BXYC)} \le &3 h(XY) + 3h(AX) + 3h(AY) +h(BX) + h(BY) + 5h(C) \nonumber\\
	& + h(XYC|AB) + \bm{4h(BC|AXY)} + h(A'C|BXY) \nonumber \\&+\bm{ h(BXY|AC)} 
	+ \bm{2h(A'BY|XC)} +  \bm{ 2h(A'BX|YC) }.\nonumber
\end{align}
By Lemma~\ref{lem:Atoa},  each of these substitutions can be balanced by adding the terms  $h(B|A'B'XYC)$ and $h(B|AB'XYC)$ respectively. 
By Lemma~\ref{lem:Atoa-otherside} we have $h(B|A'B'XYC)\ge h(B|AB'XYC)$; it is important that  the biggest  ``compensation terms'' $h(B|A'B'XYC)$ 
appears on the left-hand side of the inequality.
So again, summing the coefficients, we conclude that that the overall decrease on the left-hand side is at least that on the right-hand side, so we get
\begin{align}\label{eq:suciu3}
	11 h(A'B'XYC) \le &3 h(XY) + 3h(AX) + 3h(AY) +h(BX) + h(BY) + 5h(C) \nonumber \\
	& + h(XYC|AB) + 4h(B'C|AXY) + h(A'C|BXY) \nonumber \\&+ h(B'XY|AC)+ 2h(A'B'Y|XC) + 2h(A'B'X|YC).
\end{align}
In the \emph{third step} 
we substitute $X$ by $X'$ in the terms $h(A'B'XYC)$, $h(B'XY|AC)$, $h(A'B'X|YC)$,
\begin{align}
	\bm{11 h(A'B'XYC)} \le &3 h(XY) + 3h(AX) + 3h(AY) +h(BX) + h(BY) + 5h(C) \nonumber \\
	& + h(XYC|AB) + 4h(B'C|AXY) + h(A'C|BXY) \nonumber \\&
	+ \bm{h(B'XY|AC)}+ 2h(A'B'Y|XC) + \bm{2h(A'B'X|YC)}. \nonumber
\end{align}
Each of these three substitutions can be balanced by adding the terms $h(X|A'B'X'YC)$ (with a factor of $11$), $h(X|AB'X'YC)$ (with a factor of $1$), and again $h(X|A'B'X'YC)$ (with a factor of $2$)
respectively.
Among all these added terms,  the quantity of the entropy $h(X|A'B'X'YC)$ (appearing in  the left-hand side) is the biggest. 
This implies 
\begin{align}
	11 h(A'B'X'YC) \le &3 h(XY) + 3h(AX) + 3h(AY) +h(BX) + h(BY) + 5h(C) \nonumber \\
	& + h(XYC|AB) + 4h(B'C|AXY) + h(A'C|BXY) \nonumber \\&+ h(B'X'Y|AC)+ 2h(A'B'Y|XC) + 2h(A'B'X'|YC).\label{eq:suciu4} 
\end{align}

In the final step, we now substitute $Y$ by $Y'$ in $h(A'B'X'YC)$, $h(B'X'Y|AC)$, and $h(A'B'Y|XC)$. 
Similarly to the previous arguments, we obtain (\ref{eq:suciu-modified}), which completes the proof.

\section{A Tight Case}

Let $Q(\vset X)$ a conjunctive query and $\Sigma$ a set of functional dependencies. We say that $\vset X$ is a prefix with respect to $\Sigma$ if there is a topological order
of the dependency graph $G_\Sigma$ (see p.~\pageref{directed-graph}) in which all variables in $\vset X$ come before all variables not in $\vset X$. We show here that if $\vset X$ is a prefix with respect to $\Sigma$, then the polymatroid bound is tight.

\begin{theorem}\label{prop:prefixtight}
 Let $Q(\vset X)$ be a conjunctive query in variables $\vset Y$, let $\Pi$ be a set of cardinality constraints and $\Sigma$ be a set of functional dependencies such that $\vset X$ is a prefix with respect to $\Sigma$. Then
\begin{align*}
	\logb_{\Gamma^*_{\vset Y}}(\vset X, \Pi, \Sigma) = \logb_{\Gamma_{\vset Y}}(\vset X, \Pi, \Sigma).
\end{align*} 
\end{theorem}

In the remainder of this section, we prove \cref{prop:prefixtight}. To this end, we introduce some more definitions.
Given that  
\(
Q(\vset X) \datarule  R_1(\vset X_1), \ldots,R_\ell (\vset X_\ell)  
\)
is  a query in variables $\vset Y$, we define its \emph{output restriction} 
\(
Q^r(\vset X)  \datarule   R^r_1(\vset X_1\cap \vset X), \ldots, R^r_\ell(\vset X_\ell\cap \vset X).
\)
\begin{remark}
The transformation of $Q$ into $Q^r$ is purely syntactical, we do not discuss the semantics of the atoms $R^r_j$.  
\end{remark}

We also define $\Pi^r$ to be the set of cardinality constraints that for every $R_i^r$ is the same as for $R_i$ in $\Pi$. Finally, $\Sigma^r$ contains for every FD $\vset W \rightarrow \vset U$ in $\Sigma$ with $\vset U\cap \vset X\ne \emptyset$
\begin{example}
	Consider again the query and the FDs from Example~\ref{ex:acyclic}.
	Since in the order given there the output variables come before all non-output variables, they are a prefix with respect to the FDs. Moreover, we have 
	\begin{align*}
	Q^r(X_1,X_2, X_3) \datarule R_1^r(X_1, X_2),R_2^r(X_1),R_3^r(X_2, X_3),
\end{align*}
and the corresponding FDs are $X_1\rightarrow X_2, X_2\rightarrow X_3$. Note that the FD $X_1Y_2\rightarrow Y_3$ does not induce any FD for $Q^r$.
\end{example}

The next lemma states that if $\vset X$ is a prefix with respect to $\Sigma$, then the bounds for $Q(\vset X), \Pi, \Sigma$ and for $Q^r(\vset X), \Pi^r, \Sigma^r$ coincide.

\begin{lemma}\label{lem:tightproject}
	Let $\vset Y$ be a set of variables, $\vset X\subseteq \vset Y$, $\Pi$ a set of cardinality constraints and $\Sigma$ a set of functional dependencies such that $\vset X$ is a prefix with respect to $\Sigma$. Then
	\begin{align*}
		\logb_{\Gamma^*_{\vset Y}}(\vset X, \Pi, \Sigma) &= \logb_{\Gamma^*_{\vset X}}(\vset X, \Pi^r, \Sigma^r),\\
		\logb_{\Gamma_{\vset Y}}(\vset X, \Pi, \Sigma) &= \logb_{\Gamma_{\vset X}}(\vset X, \Pi^r, \Sigma^r).
	\end{align*}
\end{lemma}
\begin{proof}
\emph{Claim 1:}	 $\logb_{\Gamma_{\vset Y}}(\vset X, \Pi, \Sigma) \ge \logb_{\Gamma_{\vset X}}(\vset X, \Pi^r, \Sigma^r)$. 

Consider any polymatroid $h$ on $\vset X$ satisfying $\Pi^r$ and $\Sigma^r$. We will construct a polymatroid $h'$ that satisfies $\Pi$ and $\Sigma$ with $h(\vset X)= h'(\vset X)$. From this, the claim follows directly. To construct $h'$, we extend $h$ to a polymatroid on $\vset Y$ by setting for every $\vset A\subseteq Y$ the value $h'(\vset A) := h(\vset A\cap \vset X)$. Clearly, $h'$ is a polymatroid satisfying $\Pi$. We show that it also satisfies the constraints in $\Sigma$. To this end, consider an FD $\vset W\rightarrow \vset U$ from $\Sigma$. We claim that
\[
h'(\vset U\cup \vset W) = h((\vset U\cup \vset W)\cap \vset X) = h((\vset U\cap \vset X)\cup (\vset W\cap \vset X)) = h(\vset W\cap \vset X) = h'(\vset W).
\] 
In this chain of equalities only the third one is non-trivial. In the case $\vset U \cap \vset X = \emptyset$ it is straightforward; 
and in the other case it
holds because $\vset W \cap \vset X \rightarrow \vset U\cap \vset X$ is an FD in $\Sigma^r$ and thus satisfied by $h$. 
Since $h'(\vset X) = h(\vset X)$, it follows that $\logb_{\Gamma_{\vset Y}}(\vset X, \Pi, \Sigma) \ge \logb_{\Gamma_{\vset X}}(\vset X, \Pi^r, \Sigma^r)$, as claimed.

\smallskip
	 
\emph{Claim 2:} $\logb_{\Gamma^*_{\vset Y}}(\vset X, \Pi, \Sigma) \ge \logb_{\Gamma^*_{\vset X}}(\vset X, \Pi^r, \Sigma^r)$. 

To prove this claim, consider any entropic function $h$ on $\vset X$ that satisfies all constraints in $\Pi^r$ and $\Sigma^r$. We construct a new entropic function $h'$ on $\vset Y$ by adding random variables for $\vset Y\setminus \vset X$ and letting their distribution be such that they only ever take a fixed value. Then for any subset $\vset A\subseteq \vset Y\setminus \vset X$ we have $h'(\vset A) = H(\vset Y_{\vset A}) = 0$ by definition of the entropy $H$. Since $h'$ is entropic and thus a polymatroid, it follows that for every $\vset A\subseteq \vset Y$ that $h'(\vset A\cap \vset X) \le h'(\vset A) \le h'(\vset A\cap \vset X) + h'(\vset A\setminus \vset X) = h'(\vset A\cap \vset X)$, so $h'(\vset A) = h'(\vset A \cap \vset X) = h(\vset A\cap \vset X)$. Now the claim follows exactly as for the polymatroid case.
	
\smallskip
	 
\emph{Claim 3:}
$\logb_{\Gamma_{\vset Y}}(\vset X, \Pi, \Sigma) \le \logb_{\Gamma_{\vset X}}(\vset X, \Pi^r, \Sigma^r)$.

To prove this claim, we consider a polymatroid $h\in \Gamma_{\vset Y}$ satisfying the constraints in $\Pi$ and $\Sigma$. We define $h': \vset Y\rightarrow \mathbb{R}$ by setting  $h'(\vset A) = h(\vset A\cap \vset X)$ for every $\vset A\subseteq \vset Y$.
We claim that $h'$ defined this way is a polymatroid.)
Clearly, $h'(\emptyset) = 0$ and $h'$ is monotone, so it only remains to show that $h'$ is submodular. So let $\vset A, \vset B\subseteq \vset Y$. Then 
	\begin{align*}
		h'(\vset A \cap \vset B) + h'(\vset A \cup \vset B) & = h((\vset A \cap \vset B)\cap \vset X) + h((\vset A \cup \vset B)\cap \vset X)\\
		&= h((\vset A \cap \vset X)\cap (\vset B\cap \vset X)) + h((\vset A \cap \vset X) \cup (\vset B\cap \vset X))\\
		&\le h(\vset A \cap \vset X) + h(\vset B\cap \vset X)%
		= h'(\vset A) + h'(\vset B),
	\end{align*}
	which shows that $h'$ is indeed a polymatroid.
	
	We now claim that $h'$ satisfies the FDs in $\Sigma'$. So consider an FD $\vset W\cap \vset X \rightarrow \vset U\cap \vset X$ from $\Sigma'$. Assume first that $\vset W\subseteq \vset X$. Then 
	\begin{align*}
		h'(\vset W\vset U) - h'(\vset W) & = h((\vset W \cup \vset U)\cap \vset X) - h(\vset W)%
		\le h(\vset W\cup \vset U) - h(\vset W)%
		= 0,
	\end{align*}
	where the inequality holds because $h$ is monotone and the final equality it true because $h$ satisfies $\Sigma$ and thus $\vset W \rightarrow \vset U$. Now consider the case that $\vset W \nsubseteq \vset X$. Then, because $\vset X$ is a prefix, we have that $\vset U \cap \vset X = \emptyset$. It follows that 
	\begin{align*}
	h'(\vset W\vset U) - h'(\vset W) & = h((\vset W \cup \vset U)\cap \vset X) - h(\vset W\cap \vset X)%
	= h(\vset W \cap \vset X) - h(\vset W\cap \vset X)%
	= 0.
	\end{align*}
	So in any case $h'$ satisfies $\vset W\cap \vset X \rightarrow \vset U \cap \vset X$ and thus $\Sigma'$.
	
	We now define a set function $h''$ on $\vset X$ by restricting $h'$ to the ground set $\vset X$. Then $h''$ is a polymatroid that satisfies the constraints in $\Sigma'$. Moreover, we have $h''(\vset X) = h'(\vset X) = h(\vset X)$. So for every polymatroid $h$ on $\vset X$ respecting $\Sigma$, we can construct a polymatroid $h''$ on $\vset X$ respecting $\Sigma'$ such that $h''(\vset X) = h(\vset X)$, and it follows that $\logb_{\Gamma_{\vset Y}}(\vset X, \Pi, \Sigma) \le \logb_{\Gamma_{\vset X}}(\vset X, \Pi^r, \Sigma^r)$.

\smallskip
	 
\emph{Claim 4:}	
	 $\logb_{\Gamma^*_{\vset Y}}(\vset X, \Pi, \Sigma) \le \logb_{\Gamma^*_{\vset X}}(\vset X, \Pi^r, \Sigma^r)$.
	 
	 To prove this claim, we  let $h$ be an entropic function on $\vset Y$. Assume for simplicity that $\vset Y = \{1, \ldots, n\}$. Then there are random variables $\{X_1 , \ldots, X_n\}$ such that for every $\mathcal A\subseteq \vset Y$ we have $h(A) = H(X_{\vset A})$ where $X_{\vset A} = \{X_i\mid i\in \vset A\}$. Then define a set of random variables $\{X_1', \ldots, X_n'\}$ having the same distribution for the $\{X'_i\mid i\in \vset X\}$ as for $\{X_i\mid i\in \vset X\}$ and fixing all other variables to a constant value. Let $h'$ be the entropic function induced by the set $\{X_1', \ldots, X_n' \}$. Then we have for every $\vset B \subseteq \vset Y \setminus\vset X$ that $h'(\vset B) = 0$, since the corresponding random variables are constant. Since $h'$ is entropic and thus a polymatroid, it follows that for every $\vset A \subseteq \vset Y$ we have $h'(\vset A \cap \vset X) \le h'(\vset A) \le  h'(\vset A \cap \vset X) + h'(\vset A \setminus \vset X) \le h'(\vset A \cap \vset X)$ and thus $h'(\vset A) = h'(\vset A \cap \vset X)$. However, by definition of $h'$, we have $h'(\vset A\cap \vset X) = h(\vset A\cap \vset X)$ since both values are the entropy of the same random variable $X_{\vset A\cap \vset X}$. But then $h'(\vset A) = h(\vset A\cap \vset X)$. Thus, $h'$ is exactly as in 
	 Claim~3 (the case of polymatroids) but here has the additional property that it is entropic. The rest of the argument follows as the polymatroid case.
\end{proof}

To prove Theorem~\ref{prop:prefixtight}, we now use the fact that without projections the polymatroid bound and the entropic bound coincide for acyclic sets of FDs. The following is stated as Proposition~4.4~\cite{ngo2018worstb}, adapted to our slightly different notation.
\begin{proposition}\label{prop:ngo}
Let $Q(\vset X)$ be a conjunctive query without projection in variables $\vset X$, let $\Pi$ be a set of cardinality constraints and $\Sigma$ an acyclic set of functional dependencies. Then
\begin{align*}
	\logb_{\Gamma^*_{\vset X}}(\vset X, \Pi, \Sigma) = \logb_{\Gamma_{\vset X}}(\vset X, \Pi, \Sigma).
\end{align*} 
\end{proposition}

Proposition~\ref{prop:prefixtight} now follows easily.

\begin{proof}[Proof of Proposition~\ref{prop:prefixtight}]
	Simply combine the two equalities from Lemma~\ref{lem:tightproject} with that of Proposition~\ref{prop:ngo}.
\end{proof}

\section{Conclusion}

In this paper, we studied the relationship between the polymatroid and entropic bounds for conjunctive queries under acyclic functional dependencies. 
We showed that adding projections fundamentally change the picture: unlike the case of queries without projected variables~\cite{ngo2018worstb}, the two bounds do not coincide in general when projections are allowed. 
On the positive side, we identified a natural class of queries for which tightness is preserved, namely those whose output variables form a prefix with respect to the dependency graph. 
These results provide a better understanding of when the polymatroid bound remains tight.

\bibliographystyle{plain}

\begin{thebibliography}{10}

\bibitem{abo2016computing}
Mahmoud Abo~Khamis, Hung~Q Ngo, and Dan Suciu.
\newblock Computing join queries with functional dependencies.
\newblock In {\em Proceedings of the 35th ACM SIGMOD-SIGACT-SIGAI Symposium on
  Principles of Database Systems}, pages 327--342, 2016.

\bibitem{applegate2007exact}
David~L Applegate, William Cook, Sanjeeb Dash, and Daniel~G Espinoza.
\newblock Exact solutions to linear programming problems.
\newblock {\em Operations Research Letters}, 35(6):693--699, 2007.

\bibitem{ArenasBLMP21}
Marcelo Arenas, Pablo Barcel\'o, Leonid Libkin, Wim Martens, and Andreas
  Pieris.
\newblock {\em Database Theory}.
\newblock Open access at \url{https://github.com/pdm-book/community}, 2022.
\newblock accessed March 2026, commit 9f403e4f8bb14eccca301eda2feafe90179b98df.

\bibitem{atserias2013size}
Albert Atserias, Martin Grohe, and D{\'a}niel Marx.
\newblock Size bounds and query plans for relational joins.
\newblock {\em SIAM Journal on Computing}, 42(4):1737--1767, 2013.

\bibitem{BonifatiMT20}
Angela Bonifati, Wim Martens, and Thomas Timm.
\newblock An analytical study of large {SPARQL} query logs.
\newblock {\em {VLDB} J.}, 29(2-3):655--679, 2020.

\bibitem{CapelliIS25}
Florent Capelli, Oliver Irwin, and Sylvain Salvati.
\newblock A simple algorithm for worst case optimal join and sampling.
\newblock In Sudeepa Roy and Ahmet Kara, editors, {\em 28th International
  Conference on Database Theory, {ICDT} 2025, Barcelona, Spain, March 25-28,
  2025}, volume 328 of {\em LIPIcs}, pages 23:1--23:19. Schloss Dagstuhl -
  Leibniz-Zentrum f{\"{u}}r Informatik, 2025.

\bibitem{chan2011recent}
Terence Chan.
\newblock Recent progresses in characterising information inequalities.
\newblock {\em Entropy}, 13(2):379--401, 2011.

\bibitem{csirmaz2025exploring}
L{\'a}szl{\'o} Csirmaz.
\newblock Exploring the entropic region.
\newblock {\em arXiv preprint arXiv:2509.12439}, 2025.

\bibitem{dougherty2011non}
Randall Dougherty, Chris Freiling, and Kenneth Zeger.
\newblock Non-shannon information inequalities in four random variables.
\newblock {\em arXiv preprint arXiv:1104.3602}, 2011.

\bibitem{GogaczT17}
Tomasz Gogacz and Szymon Torunczyk.
\newblock Entropy bounds for conjunctive queries with functional dependencies.
\newblock In Michael Benedikt and Giorgio Orsi, editors, {\em 20th
  International Conference on Database Theory, {ICDT} 2017, Venice, Italy,
  March 21-24, 2017}, volume~68 of {\em LIPIcs}, pages 15:1--15:17. Schloss
  Dagstuhl - Leibniz-Zentrum f{\"{u}}r Informatik, 2017.

\bibitem{GottlobLVV12}
Georg Gottlob, Stephanie~Tien Lee, Gregory Valiant, and Paul Valiant.
\newblock Size and treewidth bounds for conjunctive queries.
\newblock {\em J. {ACM}}, 59(3):16:1--16:35, 2012.

\bibitem{GottlobLS02}
Georg Gottlob, Nicola Leone, and Francesco Scarcello.
\newblock Hypertree decompositions and tractable queries.
\newblock {\em J. Comput. Syst. Sci.}, 64(3):579--627, 2002.

\bibitem{GroheM14}
Martin Grohe and D{\'{a}}niel Marx.
\newblock Constraint solving via fractional edge covers.
\newblock {\em {ACM} Trans. Algorithms}, 11(1):4:1--4:20, 2014.

\bibitem{Khamis0S16}
Mahmoud~Abo Khamis, Hung~Q. Ngo, and Dan Suciu.
\newblock What do shannon-type inequalities, submodular width, and disjunctive
  datalog have to do with one another?
\newblock {\em CoRR}, abs/1612.02503, 2016.

\bibitem{Khamis0S17}
Mahmoud~Abo Khamis, Hung~Q. Ngo, and Dan Suciu.
\newblock What do shannon-type inequalities, submodular width, and disjunctive
  datalog have to do with one another?
\newblock In Emanuel Sallinger, Jan~Van den Bussche, and Floris Geerts,
  editors, {\em Proceedings of the 36th {ACM} {SIGMOD-SIGACT-SIGAI} Symposium
  on Principles of Database Systems, {PODS} 2017, Chicago, IL, USA, May 14-19,
  2017}, pages 429--444. {ACM}, 2017.

\bibitem{laszlo2026information}
Csirmaz Laszlo and Elod~P Csirmaz.
\newblock Information inequalities for five random variables.
\newblock {\em Computation}, 14(2):42, 2026.

\bibitem{leis2015good}
Viktor Leis, Andrey Gubichev, Atanas Mirchev, Peter Boncz, Alfons Kemper, and
  Thomas Neumann.
\newblock How good are query optimizers, really?
\newblock {\em Proceedings of the VLDB Endowment}, 9(3):204--215, 2015.

\bibitem{makarychev2002new}
Konstantin Makarychev, Yury Makarychev, Andrei Romashchenko, and Nikolai
  Vereshchagin.
\newblock A new class of non-shannon-type inequalities for entropies.
\newblock {\em Communications in Information and Systems}, 2(2):147--166, 2002.

\bibitem{matus2007infinitely}
Frantisek Matus.
\newblock Infinitely many information inequalities.
\newblock In {\em 2007 IEEE International Symposium on Information Theory},
  pages 41--44. IEEE, 2007.

\bibitem{matuvs2016entropy}
Franti{\v{s}}ek Mat{\'u}{\v{s}} and L{\'a}szlo Csirmaz.
\newblock Entropy region and convolution.
\newblock {\em IEEE Transactions on Information Theory}, 62(11):6007--6018,
  2016.

\bibitem{ngo2018worstb}
Hung~Q Ngo.
\newblock Worst-case optimal join algorithms: Techniques, results, and open
  problems.
\newblock In {\em Proceedings of the 37th ACM SIGMOD-SIGACT-SIGAI Symposium on
  Principles of Database Systems}, pages 111--124, 2018.

\bibitem{ngo2018worst}
Hung~Q Ngo, Ely Porat, Christopher R{\'e}, and Atri Rudra.
\newblock Worst-case optimal join algorithms.
\newblock {\em Journal of the ACM (JACM)}, 65(3):1--40, 2018.

\bibitem{selinger1979access}
P~Griffiths Selinger, Morton~M Astrahan, Donald~D Chamberlin, Raymond~A Lorie,
  and Thomas~G Price.
\newblock Access path selection in a relational database management system.
\newblock In {\em Proceedings of the 1979 ACM SIGMOD international conference
  on Management of data}, pages 23--34, 1979.

\bibitem{Shannon48}
Claude Shannon.
\newblock A mathematical theory of information.
\newblock {\em Bell Sys. Tech. Journal}, 1948.

\bibitem{Suciu23}
Dan Suciu.
\newblock Applications of information inequalities to database theory problems.
\newblock In {\em 38th Annual {ACM/IEEE} Symposium on Logic in Computer
  Science, {LICS} 2023, Boston, MA, USA, June 26-29, 2023}, pages 1--30.
  {IEEE}, 2023.

\bibitem{WangT26}
Ru~Wang and Yufei Tao.
\newblock Join and subgraph sampling under degree constraints.
\newblock {\em J. Comput. Syst. Sci.}, 155:103693, 2026.

\bibitem{Yeung12}
Raymond~W Yeung.
\newblock {\em A first course in information theory}.
\newblock Springer Science \& Business Media, 2012.

\bibitem{yeung2015facets}
Raymond~W Yeung.
\newblock Facets of entropy.
\newblock {\em Commun. Inf. Syst.}, 15(1):87--117, 2015.

\bibitem{yeung2003first}
RW~Yeung.
\newblock A first course in information theory.
\newblock {\em IEEE Transactions on Information Theory}, 49(7):1869, 2003.

\bibitem{ZhangMKOS25}
Haozhe Zhang, Christoph Mayer, Mahmoud~Abo Khamis, Dan Olteanu, and Dan Suciu.
\newblock Lpbound: Pessimistic cardinality estimation using
  {\(\mathscr{l}\)}\({}_{\mbox{p}}\)-norms of degree sequences.
\newblock {\em Proc. {ACM} Manag. Data}, 3(3):184:1--184:27, 2025.

\bibitem{zhang1998characterization}
Zhen Zhang and Raymond~W Yeung.
\newblock On characterization of entropy function via information inequalities.
\newblock {\em IEEE transactions on information theory}, 44(4):1440--1452,
  1998.

\end{thebibliography}

\appendix

\section{Appendix: computer-assisted proof of the main technical result}
\label{sct:appendix}

In this section we sketch an approach that allows us to verify with the help of a computer the correctness of our main technical result (the first inequality of Proposition~\ref{prop:concretebounds}, 
which in turn implies Theorem~\ref{thm:acyclicgap}). This provides a shortcut to the final result, allowing the reader to avoid carefully reading the details of the proof in Section~\ref{sct:entropicbound}. 
The approach involves extensive routine computations and is therefore impractical to carry out manually. Its main drawback is that it requires trusting the reliability of the computer hardware and software. 
Moreover, it gives the final answer without providing much insight into the structure underlying the computations.

Our goal is to prove a nontrivial upper bound on $\logb_{\Gamma^*_{\vset Y}}(\vset X, \Pi, \Sigma)$ associated with the query \eqref{eq:basequery-prime}
under the functional constraints \eqref{eq:untangled} and \eqref{eq:unaryUntangle} and the cardinality constraints \eqref{eq:cczy}.
To this end, we transform this information-theoretic problem into a problem of linear optimization, and then solve it with the standard methods of linear programming.

The query \eqref{eq:basequery-prime} involves  $n=9$ variables: $A,B,X,Y,C,A',B',X',Y'$. 
Thus, in  the quantity $\logb_{\Gamma^*_{\vset Y}}(\vset X, \Pi, \Sigma)$ we deal with $9$ jointly distributed random variables (which we denote by the same symbols $A,B,X,Y,C,A',B',X',Y'$).
For a $9$-tuple of random variables, we have $9$ entropy quantities for individual random variables  
\[
H(A),\ H(B),\ H(X),\ H(Y),\ H(C),\ H(A'),\ H(B'),\ H(X'),\ H(Y'),
\]
${9\choose2} = \frac{9\cdot 8}2$ entropy values for pairs
\[
H(AB), H(AC), \ldots, H(X'Y'),
\]
${9\choose3}$ entropy values for all triples, etc., up to the entropy of the entire $9$-tuple, 
\[
H(AA'BB'CXX'YY').
\]
In total, we obtain $2^9-1= 511$ entropy quantities. 
These quantities must satisfy the $6$ constraints from \eqref{eq:untangled} 
\begin{itemize}
\item $H(ABCXY) - H(AB) \le 0$,
\item $H(AB'CXY) - H(AXY) \le 0$,
\item $H(A'BCXY) - H(BXY) \le 0$,
\item $H(AB'X'Y'C) - H(AC) \le 0$,
\item $H(A'B'CXY') - H(XC) \le 0$,
\item $H(A'B'CX'Y) - H(YC) \le 0$,
\end{itemize}
and the $4$ constraints from \eqref{eq:unaryUntangle} 
\begin{itemize}
\item $H(AA') - H(A) \le 0$,
\item $H(BB') - H(B) \le 0$,
\item $H(XX') - H(X) \le 0$,
\item $H(YY') - H(Y) \le 0$.
\end{itemize}
The cardinality constraints \eqref{eq:cczy} 
can be reformulated as 
\begin{itemize}
\item $H(XY)  \le  3t$,
\item $H(AX)  \le  3t$,
\item $H(AY)  \le  3t$,
\item $H(BX)  \le  3t$,
\item $H(BY)  \le  3t$,
\item $H(C)  \le  2t$,
\end{itemize}
where $t$ stands for the threshold $\log N$.
Furthermore, the entropy quantities under consideration are non-negative and satisfy the axioms of polymatroids (monotonicity, subadditivity, and submodularity, also known as \emph{Shannon inequalities}).
For $n=9$ random variables, all these inequalities are implied by  
\[
n + {n\choose 2} 2^{n-2}  = 9 + \frac{9\cdot 8}2 \cdot 2^7 = 4617
\]
 \emph{elemental} Shannon inequalities, see, e.g.,  \cite{Yeung12}.
The description of the elemental inequalities is explicit, and the complete list of these inequalities can be generated mechanically by a simple computer program.

In addition to the constraints above, we add one more inequality
\begin{align*}
	11 H(A'B'X'Y'C) \le &3 H(XY) + 3H(AX) + 3H(AY) +H(BX) + H(BY) + 5H(C) \\
	& + H(ABCXY) - H(AB) \\
	& + 4H(AB'CXY)  - H(AXY)\\
	& + H(A'BCXY) - H(BXY) \\
	&+ H(AB'CX'Y') - H(AC)\\
	& + 2H(A'B'CXY')  - H(XC)\\
	&  + 2H(A'B'CX'Y) - H(YC),
\end{align*}
which is a form of the non-Shannon type inequality (\ref{eq:suciu-modified}).

Given all those constraints, we maximize the entropy $H(A'B'CXY)$.
This is essentially an optimization problem with linear constraints, i.e., a linear programming problem
(with one  optimization variable  for each of the $511$ entropy quantities).
It turns out that the optimal value of the resulting linear program is
\(
 \frac{43}{11} t.
\)
This answer can be found up to floating-point rounding errors with standard industrial LP solvers, 
or \emph{exactly} with an LP solver that 
computes rational optimal solutions
based on the simplex algorithm implemented with arbitrary-precision rational arithmetic
(see \cite{applegate2007exact}).
Although the resulting linear program is far too large to solve manually, it is can be solved easily by a computer with modern LP solvers.

\end{document}